\documentclass[12pt,reqno]{amsart}
\usepackage{amssymb}
\usepackage{amscd}
\usepackage{amsxtra}
\usepackage[mathscr]{eucal}
\usepackage[all]{xy}
\usepackage{amsmath}
\newtheorem{mydef}{Definition}
\newtheorem{mythm}{Theorem}
\newtheorem{mypro}{Proposition}
\newtheorem{mylem}{Lemma}
\newtheorem{myconj}{Conjecture}

\usepackage[utf8]{inputenc}
\usepackage[english]{babel}
\usepackage[backend=biber,style=alphabetic,sorting=ynt]{biblatex}
\usepackage{tikz}

\newcommand{\unit}{\text{\textbf{1}}}

\newcommand{\id}{\mbox{\rm id\,}}

\theoremstyle{plain}
\numberwithin{equation}{section}

\theoremstyle{definition}

\theoremstyle{remark}

\newcommand{\mcA}{\mathcal{A}}
\newcommand{\mcB}{\mathcal{B}}
\newcommand{\mcC}{\mathcal{C}}

\newcommand{\mcH}{\mathcal{H}}

\newcommand{\mcO}{\mathcal{O}}
\newcommand{\mcP}{\mathcal{P}}

\newcommand{\mcV}{\mathcal{V}}

\newcommand{\mbbC}{\mathbb{C}}

\newcommand{\mbbR}{\mathbb{R}}
\newcommand{\mbbS}{\mathbb{S}}

\newcommand{\mbbZ}{\mathbb{Z}}

\author{Yang Qiu}
\address{Department of Mathematics, University of California, Santa Barbara, CA 93106, USA}
\email{yangqiu@math.ucsb.edu}

\author{Zhenghan Wang}
\address{Microsoft Station Q and Department of Mathematics, University of California, Santa Barbara, CA 93106, USA}
\email{zhenghwa@microsoft.com;zhenghwa@math.ucsb.edu}

\begin{document}

\title[Error Correction and disk axiom]{Ground subspaces of topological phases of matter as error correcting codes}

\thanks{Z.W. is partially supported by NSF grant
FRG-1664351 and DOD Muir grant ARO W911NF-19-S-0008.  We thank Shawn Cui for insightful discussions, especially ideas on how to adapt the proof of TQO1 to cover TQO2. The second author thanks J. Haah for helpful discussion and communication.}

\begin{abstract}

Topological quantum computing is believed to be inherently fault-tolerant.  One mathematical justification would be to prove that ground subspaces or ground state manifolds of topological phases of matter behave as error correction codes with macroscopic distance.  While this is widely assumed and used as a definition of topological phases of matter in the physics literature, besides the doubled abelian anyon models in Kitaev’s seminal paper, no non-abelian models are proven to be so mathematically until recently.  Cui et al extended the theorem from doubled abelian anyon models to all Kitaev models based on any finite group \cite{Cui19}.  Those proofs are very explicit using detailed knowledge of the Hamiltonians, so it seems to be hard to further extend the proof to cover other models such as the Levin-Wen.  We pursue a totally different approach based on topological quantum field theories (TQFTs), and prove that a lattice implementation of the disk axiom and annulus axiom in TQFTs as essentially the equivalence of TQO1 and TQO2 conditions in \cite{BHM10,BH11}.  We confirm the error correcting properties of ground subspaces for topological lattice Hamiltonian schemas of the Levin-Wen model and Dijkgraaf-Witten TQFTs by providing a lattice version of the disk axiom and annulus of the underlying TQFTs.  The error correcting property of ground subspaces is also shared by gapped fracton models such as the Haah codes \cite{Haah11,Haah13}.  We propose to characterize topological phases of matter via error correcting properties, and refer to  gapped fracton models as lax-topological.

\end{abstract}


\maketitle

\section{Introduction}

Experimental quantum computing has enjoyed tremendous success in the last three decades that quantum super-advantage has been published \cite{AQS19}.  So far the focus has been on the construction of high quality physical qubits: those directly out of quantum systems without any error correction.  For scalable universal quantum computing, logical qubits seems inevitable: those embedded in physical qubits according to error correction codes \cite{Shor95}.  Topological quantum computing is an alternative approach to fault-tolerant scalable universal quantum computing using topological phases of matter.  To distinguish the two approaches, we will refer to the two-step---first physical qubits, then logical qubits---as traditional or circuit quantum computing.  In both approaches, topological quantum error correction codes will play an essential role.  The theoretical architecture of error correction in surface codes for superconducting qubits is based on boundary defects between the two different gapped boundaries of the toric code, which are essentially equivalent to the Ising anyons or Majorana zero modes \cite{Fowler+12,BK98,CMW16}.  On the other hand, the promise of inherent fault-tolerance in topological quantum computing could be justified by the error correction properties of the ground states.

That the ground subspaces or ground state manifolds of topological phases of matter behave as error correction codes with macroscopic distance is widely assumed and used as a definition of topological phases of matter in the physics literature \cite{BN13}.  But besides the doubled abelian anyon models in Kitaev’s seminal paper, no non-abelian models are proven to be so mathematically until recently.  Cui et al extended the theorem from doubled abelian anyon models to all Kitaev models based on any finite group \cite{Cui19}.  Those proofs are very explicit using detailed knowledge of the Hamiltonians, so it seems to be hard to further extend the proof to cover other models such as the Levin-Wen.  We pursue a totally different approach based on topological quantum field theories (TQFTs), and prove that a lattice implementation of the disk axiom and annulus axiom in TQFTs as essentially the equivalence of TQO1 and TQO2 conditions in \cite{BHM10,BH11}.  We confirm the error correcting properties of ground subspaces for topological lattice Hamiltonian schemas of the Levin-Wen model and Dijkgraaf-Witten TQFTs by providing a lattice version of the disk axiom and annulus of the underlying TQFTs.  The error correcting property of ground subspaces is also shared by gapped fracton models such as the Haah codes \cite{Haah11,Haah13}.  We propose to characterize topological phases of matter via error correcting properties, and refer to gapped fracton models as lax-topological.

Topological (and lax-topological) phases are interesting for many deep reasons.  Topological phases form an exciting sub-field in condensed matter physics lying beyond Landau’s group symmetry and symmetry breaking paradigm.  Topological phases lead to a paradigm shift in physicists’
perspectives on phases: rather than symmetry breaking, more an emergence of new higher category quantum symmetries as a complimentary theme.
Theoretically, topological phases would also allow us to build a topological quantum computer, which would have wide applications.
Finally, topological phases and their topological orders are interesting mathematical objects in their own right.  The approach to topological phases could shed light on the mathematics of general quantum field theories.  The paradigm approach to topological phases is lattice models so that topological partition functions are state-sums.  Analogously to the relationship between Riemann sums and definite integrals, lattice models could be regarded as the Riemann sums for path integrals. 

In this paper, all our TQFTs and categories are unitary.  Topological order and topological phases of matter are for bosonic/spin systems, as opposed to fermionic and/or symmetry protected trivial (SPT) topological phases.

The content of the paper is as follows.  In Sec. \ref{char}, we discuss various characterizations of topological phases and fracton order.  In Sec. \ref{statesum}, we briefly outline how to go from state sum TQFTs to Hamiltonian schemas.  In Sec. \ref{LW}, we verify that the Leven-Wen model is topological in our sense and do the same for DW TQFTs in Sec. \ref{DW}.

\section{Characterization of Topological Phases of Matter}\label{char}

There are no universally adopted mathematically rigorous definitions of topological order and/or topological phase of matter yet, which reflects the active status of the field. Informally, we will adopt the following working conceptual definitions of topological order and topological phases of matter.  A topological order is a catch-all higher category of all universal properties of a topological phase of matter. A topological phase is a phase of matter whose low energy effective theory is a topological quantum field theory (TQFT).  There are definitions in the physics literature based on Hamiltonians or states, but making these into a mathematically rigorous definition is still a significant open problem with subtleties. In dimension $(2 + 1)$, unitary modular tensor categories mathematically or anyon models physically are regarded essentially as topological orders, which is completely rigorous, while a topological phase of matter is a path-connected component of a space of topologically ordered Hamiltonians or an equivalence class of certain states, though how to define the space of Hamiltonians or the equivalence relation of the states is not completely clear.  The definitions of weak higher $n$-categories have not yet converged in mathematics either.

Instead of trying for a rigorous definition, which is probably still pre-mature, we will study characterizations of topological phases that eventually would lead to rigorous mathematical definitions.  Topological phases have many facets including two equivalent characterizations in $(2+1)$ dimensions: the collection of ground subspaces consisting of ground states on all space manifolds on one hand, and the algebraic structure of all the elementary excitations in the plane on the other hand.  The ground subspaces of lattice Hamiltonians should be stable to form a topological modular functor, while the elementary excitations in the plane should form an anyon model.  Topological modular functor and anyon model are two sides of the same coin: one geometric and the other algebraic, while are equivalent and both lead to TQFTs.

The paradigm example of a topological phase of matter is the toric code, which realizes the Drinfeld center of $\mbbZ_2$ anyon model and the Dijkgraaf-Witten TQFT based on $\mbbZ_2$.  An important extension is the Haah code, more general fracton models, which will be referred to as lax-topological phases because they do not strictly fit into the conventional TQFT formalism. The toric code is well-understood---
almost every conceivable question could be answered, while Haah code goes beyond conventional TQFT, and is poorly understood categorically.  
A proper understanding of a continuum limit of Haah code would likely go beyond the framework of renormalizable Lorentian invariant quantum field theory, hence Haah code would become another paradigm example for lax-topological phases.

\subsection{Basic definitions}

We introduce definitions of quantum theory, Hamiltonian schemas, topological schema, etc below for a proper mathematical discussion following \cite{RW18}.

\begin{mydef}\label{quantum:def}

A quantum theory is a triple $(L,B,H)$, where $L$ is a finite-dimensional complex Hilbert space, $B$ an orthnormal basis of $L$, and $H$ a Hermitian operator on $L$, which would be regarded as a matrix using $B$.

\end{mydef} 

The basis $B$, which represents classical configurations, is an unusual ingredient, but is important in the study of quantum information. 
Almost all quantum theories according to Def. \ref{quantum:def} are not related to real physics since they do not satisfy physical constraints such as locality. Our focus will be on examples which come from theoretical physics, though might be still far from experimental physics.

Given a quantum theory $\mcH=(L,B,H)$, the distinct eigenvalues $\{\lambda_i\}$ of the Hamiltonian $H$, ordered as $\lambda_0<\lambda_1<\cdots $ with corresponding eigenspaces $L_{\lambda_i}$, are called energies (or energy levels) of the
theory.  The difference $\lambda_1-\lambda_0$ is called the energy gap.  The eigenspace $\lambda_0$ is called the ground state energy and usually normalized to $0$ by adding a multiple of the identity to the Hamiltonian\footnote{Energy is really an $\mbbR$-torsor,
 and can be shifted up or down by adding some multiple of the identity to the Hamiltonian, relabeling the same ground states with different
energy values as long as the energy is bounded below.}.  The nonzero eigenvectors of $L_{\lambda_0}$ are called ground states. Nonzero eigenvectors for other eigenvalues are called excited states.  We are mainly interested in the low energy physics, i.e., so the properties of the eigenstates for the first few smaller eigenvalues.

\begin{mydef}

An $n$-dimensional quantum schema $\mcH$ is a rule that assigns to every $n$-dimensional manifold $Y$ with some celluation $\Delta$ a quantum theory $\mcH(Y,\Delta)$ with Hamiltonian $H_\Delta$.

\end{mydef}

Abstractly, the manifold $Y$ is imagined as our physical space and the celluation $\Delta$ is a configuration of fundamental constituent particles such as atoms or electrons of a material in $Y$, where the particles live on vertices of $\Delta$ and the higher skeletons such as edges and faces represent the patterns of the interactions among the particles.

\begin{mydef}

An $n$-dimensional Hamiltonian schema is sharp gapped if there is a positive constant $c>0$ such that for all $n$-manifolds $Y$ and their celluations $\Delta$, all Hamitlonians $H_\Delta$ of the resulting quantum systems $\mcH(Y,\Delta)$ is frustration-free and has an energy gap $\geq c$. 

\end{mydef}

Crucially, the constant $c$ does not depend on the celluations. Sharp gapped Hamiltonian schemas are almost topological.

\begin{mydef}

An $n$-dimensional Hamiltonian schema is topological if 

\begin{enumerate}
    \item TQO0: the schema is sharp gapped,
    \item TQO1: the ground states form error correction code,
    \item TQO2: the ground states are homogeneous,
    \item TQO3: the ground state degeneracy is independent of the celluations.
\end{enumerate}

\end{mydef}

The TQO1 and TQO2 conditions are defined in \cite{BHM10,BH11}, which are sufficient conditions for the stability of the spectral gap of a sharp gapped Hamiltonian schema.  We will recall them more precisely in later sections.

It is expected that a topological Hamiltonian schema functorially defines a topological modular functor in the sense there
exists a unitary topological modular functor $V$ such that
for any closed $n$-manifold $Y$, $V(Y)$ is isomorphic to the space of ground states of the Hamiltonian schema on
$Y$ over some limit of appropriate celluations. Therefore, a topological Hamiltonian schema represents mathematically a topological modular functor, hence a topological phase of matter.

\begin{mydef}

A topological phase of matter is morally a path-connected component of the space of topologically ordered Hamiltonians.

\end{mydef}

The problem with this working definition is that there are no formal definitions of  what are the space of topologically ordered Hamiltonians and what kinds of free degrees of freedom are allowed for some entanglement renormalziation.  There are things that we anticipate such as two Hamiltonians are regarded as equivalent if there is a continuous path deforming one into the other through topological Hamiltonian schemas — in particular, all Hamiltonians on the path cannot close the gap. Understanding this definition carefully in general would require some version of the renormalization idea in quantum field theory.  Another important ingredient is stabiliation: the deformaiton of Hamiltonians is allowed to have access to trivial degrees of freedom, and depends on what kinds of local degrees of freedom is permissible.  Moreover, in reality, not all celluations should be allowed as lattices for real materials are usually highly constrained by quantum chemistry.

The toric code is a well-known topological schema, but the verification of the extension in our sense to all finite groups is done only recently in \cite{Cui19}. That the same holds for many other models are widely believed and assumed:

\begin{enumerate}
    \item The low-energy TQFTs of the Levi-Wen model should be the Turaev-Viro type TQFTs based on Drinfeld centers.

\item The low-energy TQFTs of the Walker-Wang model (which is in dimension $(3 + 1)$ should be the Crane-Yetter TQFTs based on pre-modular categories.

\item The low-energy TQFTs of the lattice model from $G$-crossed braided categories should be the $(3+1)$-TQFTs constructed by Cui from the same data.
\end{enumerate}

The main result of the paper is to verify the Levin-Wen /Turaev-Viro case.  The proof reduces other cases to the lattice implementation of certain TQFT axioms.

\subsection{Ground subspaces and topological modular functors}

\subsubsection{Hamiltonian realization of topological modular functors}

An $n$-dimensional topological modular functor (TMF) $\mcV$ is the $n$-dimensional part of a $1$-extended $(n+1)$-TQFT.  The full axioms are explicitly listed in \cite{RW18}.  For our purpose, we recall the disk axiom, annulus axiom and a consequence of the general axioms, which is called the fusion axiom below.

In our formulation, a TQFT always comes with a label set $\Pi$, which consists of the pointed topological charges of the theory\footnote{To deal with subtleties from Frobenius-Schur indicators,  we actually need to use a representative set of an appropriate category.}. The boundary components of $n$-manifolds are labeled
by elements of the label set.

\begin{enumerate}

\item \textrm{Disk axiom:}

$V(D^n;X_i)\cong \begin{cases}
\mbbC, & X_i=\unit,\\
0, & \textrm{Otherwise}, \end{cases}$
\quad \quad where $D^n$ is an $n$-disk topologically and a label $X_i\in \Pi$, where $\unit$ the tensor unit.

\item \textrm{Cylinder axiom:}

\[V(\mcA;X_i,X_j)\simeq \begin{cases} \mbbC & \textrm{if $X_i\simeq X_j^{*}$},\\ 0 & \textrm{otherwise}, \end{cases}\]
where $\mcA$ is the cylinder $\mbbS^{n-1} \times I$ topologically, and $X_i,X_j\in \Pi$.

\item \textrm{Fusion axiom:}

\[V(\mcP;X_i,X_j,X_k)\simeq \begin{cases} \mbbC^{N_{ijk}} & \textrm{if $X_i, X_j, X_k$}\;  \textrm{satisfies the fusion rules},\\ 0 & \textrm{otherwise}, \end{cases}\]
where $N_{ijk}$ is some non-negative integer and $\mcP=P\times \mbbS^{n-2}$
for a pair of pants $P$.  

\end{enumerate}

Locality of a TMF $\mcV$ in the form of a gluing formula implies that a state in $V(Y)$ can be constructed from states on local patches of $Y$.  In $(2+1)$-dimensions, every compact surface $Y$ has a DAP decomposition:  a decomposition of $Y$ into disks, annuli, and pairs of pants.  
Therefore, general anyon states on any space $Y$ can be reconstructed from anyon states on disks, annuli, and pairs of pants together with some general principles. The topology of the disk, annulus, and pair of pants detect the vacuum, anti-particles, and fusion rules.  Appropriate analogues in higher dimensions hold.

Given a Hamiltonian schema $\mcH$ and a space manifold $Y$ with a celluation $\Delta$, all ground states of a Hamitltonian $H_\Delta$ on $Y$ form a Hilbert subspace $V(Y,\Delta)$.  If the ground subspaces have some kind of limit to a Hilbert space $V(Y)$, and the limit is functorial so that the map from $Y$ to $V(Y)$ is a TMF $\mcV$, then the Hamiltonian schema $\mcH$ is called to realize the TMF $\mcV$, which means the low energy physics of $\mcH$ is a TMF.

\subsection{Elementary excitations and anyon models}

While the ground states of a topological Hamiltonian schema should realize a TMF, the low energy elementary excitations, usually the first and second excited states in $L_{\lambda_1}$ and $L_{\lambda_2}$,  should form an anyon model or a unitary modular tensor category.  The elementary excitations correspond to ones that cannot be factored further, though this is not a completely mathematical term due to constraints such as symmetries.

\subsection{The Disk axiom and error correction codes}

\subsubsection{Oneness: lattice version of the disk axiom}

In a quantum error-correcting code, there exists the large finite-dimensional Hilbert space $L$ of physical qubits, and the subspace $V$ of logical qubits where information is encoded to be protected from imperfections and errors. The subspace $V$ behaves
like a one-dimensional subspace with respect to certain local operators in the following sense.

Given a space $(\mbbC^l)^{\otimes n}$ of qudits for some integer $l>0$.
An operator $\mcO$ is $m$-local for some integer $m\geq 1, n>m$ on $(\mbbC^l)^{\otimes n}$ if $\mcO$ is of the form $\textrm{Id}\otimes A  \otimes \textrm{Id}$, where $A$ acts on $m$ qudits.  

An error-correcting code is an embedding of $(\mbbC^l)^{\otimes k}$ into $(\mbbC^l)^{\otimes n}$ such that information in the image of $(\mbbC^l)^{\otimes k}$ is protected from local errors on $(\mbbC^l)^{\otimes n}$, i.e. $m$-local operators for some $m$ on $(\mbbC^l)^{\otimes m}$ cannot change the embeded states of $(\mbbC^l)^{\otimes k}$ in $(\mbbC^l)^{\otimes n}$, called the code subspace, in an irreversible way.  One characterization of the error correction code is the following condition, called TQO1: the composition of operators of the form 

\[ \xymatrix{ (\mbbC^l)^{\otimes k} \ar@{^(->}[r]^i & (\mbbC^l)^{\otimes n} \ar[r]^{O_m} & (\mbbC^l)^{\otimes n} \ar@{->>}[r]^\pi & (\mbbC^l)^{\otimes k} }\]
is $\lambda \cdot \id$ for any $m$-local operator $O_m$ on $(\mathbb{C}^l)^{\otimes n}$ for some scalar $\lambda$, where $i$ is the inclusion and $\pi$ the projection.

When $\lambda \neq 0$, $O_m$ does not degrade the logical qubits as projectively its action on the logical qubits is the identity operator.
But when $\lambda=0$,  $O_m$ rotates logical qubits out of the code subspace, to orthogonal states, which can be detected by simultaneous measurements and subsequently corrected by using designed gates. 

\subsubsection{Topological invariance: annulus axiom}

The annulus axiom in a TMF posits that a topological state cannot change if there is no change of topology: the evolution from a product $Y\times I$ bordism is projectively the identity operator.  A direct consequence is that if $A$ is a topological disk, then an enlargement of $A$ by adding a collar region $A\times I$ will not change the topological states.

Suppose $A$ is a lattice disk region and $B$ is an expansion of $A$ via adding a collar.  For a given Hamiltonian $H$, let $V_B$ be the states that are joint eigenstates of all local terms of $H$ with support in $B$, then 
the TQO2 condition is that the kernel of the first sequence is the same as that of the second below:

\[ \xymatrix{ (\mbbC^l)^{\otimes k} \ar@{^(->}[r]^i & (\mbbC^l)^{\otimes n} \ar[r]^{O_A} & (\mbbC^l)^{\otimes n} &}\]

\[ \xymatrix{ V_B \ar@{^(->}[r]^i & (\mbbC^l)^{\otimes n} \ar[r]^{O_A} & (\mbbC^l)^{\otimes n}  &  ,}\]

where $O_A$ is any local operator supported in $A$.

The version of TQO2 that is proved in our paper later is an operator version of TQO1. More specifically, $O_AP=0$ implies $O_AP_B=0$.

\subsection{Lax-topological phase}

Recently, interests are expanded beyond anyons and anyon models such as studying symmetry defects, gapped boundaries, (2 + 1)-dimensional black holes, and fractons. Fracton order should
follow a definition similar to the topological schema with the condition that the ground state degeneracy is independent
of the cellulation replaced by some other condition.
The Hamiltonians from Haah codes \cite{Haah11} are sums over the cubes, with operators associated to each edge of the cube.
The ground subspaces are error-correcting codes \cite{Haah13}, and the ground state degeneracy depends on the lattice and is unbounded.

\begin{mydef}
A $D$-dimensional Hamiltonian schema is lax-topological if it satisfies $TQO0,TQO1,TQO2$, but $TQO3$ is replaced by the following condition: there exist constants $\alpha$ and $\beta$ depending only on $D$ such that 
$GSD(H_\Delta)\leq \alpha e^{\beta L^{D-2}}$ for any Hamiltonian $H_\Delta$ from the schema on a linear size $L$ celluation $\Delta$ of a space manifold $Y$.
\end{mydef}

Haah codes saturate the bound above, and it is an interesting open question if there are intermediate maximum growth functions for ground state degeneracy in between TQFTs and the Haah codes \cite{Flam+17}.

\begin{myconj}
All gapped fracton schemas are lax-topological.
\end{myconj}

\section{From state-sums to Hamiltonian schemas}\label{statesum}

\subsection{State-sum TQFTs}

State-sum construction of topological quantum invariants is an approach to quantum invariants using ideas from statistical mechanics.  The simplest illustration is the Euler characteristic of a simplicial complex.  Fixing a positive number $a$, and suppose $X$ is an $m$-dimensional simplicial complex with the $i$-skeleton $\Delta^i$---all the closed $i$-simplices.  If each $i$-simplex $s$ is assigned with weight $a^{(-1)^i}$, then the state sum $Z(X)$ of $X$ is defined as 

$$Z(X)=\frac{\prod_{i=0,\textrm{even}}^m (\prod_{s\in \Delta^i}a)}{\prod_{i=1,\textrm{odd}}^m (\prod_{s\in \Delta^i}a)}=a^{\chi(X)},$$

where $\chi(X)=\sum_{i=0}^m (-1)^i |\Delta^i|$ is the Euler characteristic of $X$.

While there are no formal definitions of what is a state sum, it is widely believed that a weak $n$-category $\mcB$ with some additional structures would lead to an $(n+1)$-TQFT, in particular partition functions of $(n+1)$-dimensional space-time manifolds $X$ and Hilbert spaces for  $n$-dimensional space manifolds $Y$.  For the construction of the state-sum for an $(n+1)$-manifold $X$, one usually starts with a triangulation or celluation $\Delta$ of $X$, then all skeletons in $\Delta$ are colored by data from $\mcB$.  Finally a sum $Z(X)$ is defined and proved to be independent of all choices, especially the triangulation $\Delta$.

\subsection{From state-sum to Hamiltonian schemas}

Most state-sum constructions of $(n+1)$-TQFTs in the literature focus on the partition functions or topological invariants of the $(n+1)$-manifolds.  As is widely expected, all state-sum TQFTs are fully extended so in particular  $1$-extended.  The extension to $n$-manifolds is straightforward, but the extension to dimension $n-1$ requires the description of certain categories, which are in general hard in an explicit way.  One reason for our proof of the Levin-Wen and DW is that a good understanding of the related categories is available.

Given some category $\mcB$ and a celluation $\Delta$ of an $n$-manifold $Y$.  For simplicity, the label set $\Pi_\mcC$ is assumed to consist of a complete set of representatives of simple objects of $\mcB$.  Furthermore, only the $1$ and $2$ skeletons $\Delta^{(i)},i=1,2$, are colored as in all our examples.

To define the Hilbert space $V(Y)$, consider a celluation $\Delta_{Y\times I}$ of $Y\times I$ that extends the celluation $\Delta$ of $Y$.  For any two colorings of the $1$-skeleton of $\Delta^{(1)}$, define the relative state-sum over all extensions of the two boundary colorings to be colorings of the $1$-skeleton of $\Delta_{Y\times I}$ with the convention that the state-sum is $0$ if there are no extensions.  Regarding the relative state-sum as a matrix entry, we obtain a linear map $\tilde{Z}(Y\times I)$ from the vector spaces spanned by all colorings of $\Delta^{(1)}$ to itself.  It can be shown that the map $\tilde{Z}(Y\times I)$ is an idempotent, and finally $V(Y)$ is defined to be the image of $\tilde{Z}(Y\times I)$.

Topological invariance of the state-sum leads to a gluing formula and hence a description of $\tilde{Z}(Y\times I)$ using local patches, thus a Hamiltonian schema.

\subsubsection{Explicit Hamiltonians}

To write down explicit Hamiltonians, it is convenient to define admissible colorings as those such that the Hom space around each vertex is non-zero.  Then the first family of local terms in the Hamiltonian schema simply enforces the admissibility condition around each vertex for colorings of the $1$-skeleton, physically some version of the Gauss law.  Then for each $2$-cell in the celluation, a term to enforce the $0$-flux condition through the attaching circle of each $2$-handle is added as the second family of local terms. The difficulty lies in the lack of complete symmetry in $6j$ symbols so care must be taken in the ordering of vertices of a $3$-cell.  Various solutions exist for specific models such as the Levn-Wen and Walker-Wang models.

\subsubsection{Gluing formulas and ground subspaces}

Suppose $N$ is a small neighborhood of the $1$-skelton $\Delta^{(1)}$ and $N\cup h$ the union with a $2$-handle $h$ along an attaching circle $s$.  Then the gluing axiom implies that the Hilbert space $V(N\cup h)\cong \oplus_{l} V(N;l)\otimes V(h,l)$, where $l$ are labels for the attaching region of $s\times I$.  Only the trivial label $l_0$ has non zero contribution and $V(h,l_0)\cong \mbbC$ by the disk axiom, hence $V(N\cup h)\cong \textrm{Im}(e_0)(V(N))$, where $e_0$ is a projection onto the $0$-flux subspace from the admissible subspace. It follows that $V(Y)\cong \textrm{Im}(\prod_{h}e_h(V(N)))$, where the product is over all $2$-handles.

\begin{myconj}

All state-sum Hamitlonian schemas are topological.
   
\end{myconj}

\subsection{Non-commuting Hamiltonian schemas}

Topological Hamiltonian schemas go beyond commuting ones, but mathematical techniques to prove the existence are rare.  Numerical simulations provide strong evidence that the Heisenberg anti-ferromagnetic spin Hamiltonian realizes the same TQFT as the toric code \cite{JWB}, and Haldance hardcore boson model realizes the chiral Semion model \cite{CV}.  The Heisenberg anti-ferromagnetic spin Hamiltonian is potentially realized by a real material Herbertsmithite.

\section{The Levin-Wen model}\label{LW}

It is widely believed that the Levin-Wen model realizes Turaev-Viro type TQFTs that are quantum doubles \cite{LW05}, i.e. those constructed from a unitary fusion category $\mathcal{C}$ using  triangulations of manifolds.  One version of such a statement would be that the Levin-Wen schema realizes Turaev-Viro TQFTs, and another would be that the elementary excitations or quasi-particles of the Levin-Wen model form the Drinfeld center $Z(\mathcal{C})$/quantum double $D(\mathcal{C})$ of the input $\mathcal{C}$.

A mathematical definition of the Levin-Wen model in full generality is still complicated due to the Frobenius-schur indicators.  To the best of our knowledge, there are no proofs that the Levin-Wen model in \cite{LW05} is frustration-free and realize Turaev-Viro type TQFTs in either sense above.  By general TQFT consideration, there is a version of Levin-Wen model that would be frustration-free.  We reformulate the Levin-Wen model for the application of results in \cite{Kirillov11} below. 

\subsection{The Levin-Wen model}

 Let $Y$ be a closed surface with a cell structure $\Delta$ as defined in \cite{KB10}, \cite{Kirillov11} with an orientation for each edge and a starting vertex for each 2-face $f$. The set of isomorphism classes of irreducible objects of $\mathcal{C}$ is denoted by $\pi_\mcC$, the number of edges of $f$ by $n(f)$. We use $e, f$ to denote the edges, faces of $\Delta$ respectively.
 
Define Hilbert space $L$ as:
$$L=\bigotimes_{e\in\Delta}\mathbb{C}^{\pi_{\mathcal{C}}}\otimes\bigotimes_{f\in\Delta}\text{Hom}(1,(\bigoplus_{X\in\pi_{\mathcal{C}}}X)^{\otimes n(f)}).
$$

The Hilbert space $L$ can be rewritten as
$$L=\bigotimes_{e\in\Delta}\mathbb{C}^{\pi_{\mathcal{C}}}\otimes\bigotimes_{f\in\Delta}\bigoplus_{X_1,...,X_{n(f)}\in \pi_{\mathcal{C}}}\text{Hom}(1,\bigotimes_{i=1}^{n(f)}X_i)
$$
with a basis of the form:
$$\bigotimes_{e\in\Delta}X_e\otimes\bigotimes_{f\in\Delta}(x_1^f,...,x_{|\pi_{\mathcal{C}}|^{n(f)}}^f),
$$
where $x_j^{f}$ is the basis of $\text{Hom}(1,\bigotimes_{i=1}^{n(f)}X_i)$ for some $X_1,...,X_{n(f)}$. When $X_i$ violates the fusion rules, $x_j^f$ is set to be $0$.

For any face $f_0$ of $\Delta$, we define $H_{f_0}$ on $L$ by
$$H_{f_0}(\bigotimes_{e\in\Delta}X_e\otimes\bigotimes_{f\in\Delta}(x_1^f,...,x_{|\pi_{\mathcal{C}}|^{n(f)}}^f))=\bigotimes_{e\in\Delta}X_e\otimes\bigotimes_{f\neq f_0}(x_1^f,...,x_{|\pi_{\mathcal{C}}|^{n(f)}}^f)\otimes(0,...,x_j^{f_0},...,0)
$$
where $x_j^{f_0}$ corresponds to the component for $X_1,...,X_{n(f_0)}$ matching the colorings on the edges of $f_0$.
For any vertex $v$ of $\Delta$, we define $H_v=(H_{v,loc}\otimes id)\circ\prod_{f_v}H_{f_v}$ where $f_v$ is the face containing $v$, $H_{v,loc}$ acts on $\bigoplus_{l_v}\bigotimes_{f_v}H(f_v,l_v)$, $l_v$ is the labeling on the edges of all $f_v$, $H(f_v,l_v)$ is defined in \cite{KB10} and $id\in End(\bigotimes_{e\notin f_v}\mathbb{C}^{\pi_{\mathcal{C}}}\otimes\bigotimes_{f\neq f_v}\text{Hom}(...))$. Define $H_{v,loc}$ as shown in Fig \ref{fig1}. The dashed lines represent the dual complex of $\Delta$. $y_i$ represents the sum over all dual basis which is used in \cite{Kirillov11}. The blue region represents the contraction along the red lines. Then $H_f,H_v$ are projective, $H_fH_{f^{\prime}}=H_{f^{\prime}}H_f$, $H_vH_{v^{\prime}}=H_{v^{\prime}}H_v$. Let $\tilde{V}(\Delta)$ be the large vector space in Turaev-Viro TQFT and $V(\Delta)$ be the image of $\tilde{V}(\Delta)$ under partition function $\mathcal{H}_p=Z(\Sigma\times I)$. We have $\tilde{V}(\Delta)=\prod_{f\in\Delta}H_f(L)$ and $V(\Delta)=\prod_{v\in\Delta}H_v\prod_{f\in\Delta}H_f(L)$.
\begin{figure}
    \centering
\begin{tikzpicture}[scale=0.7]

\draw (-1,0)--(7,0);
\draw (-1,-3)--(7,-3);
\draw (-1,-6)--(7,-6);
\draw (0,1)--(0,-7);
\draw (3,1)--(3,-7);
\draw (6,1)--(6,-7);

\node[circle,draw,scale=0.7] (1) at (1.5,-1.5) {$x_1$};
\node[circle,draw,scale=0.7] (2) at (1.5,-4.5) {$x_2$};
\node[circle,draw,scale=0.7] (3) at (4.5,-4.5) {$x_3$};
\node[circle,draw,scale=0.7] (4) at (4.5,-1.5) {$x_4$};

\node at (3,-1.5) {$X_1$};
\node at (1.5,-3) {$X_2$};
\node at (3,-4.5) {$X_3$};
\node at (4.5,-3) {$X_4$};

\draw[dashed] (1)--(2);
\draw[dashed] (1)--(4);
\draw[dashed] (2)--(3);
\draw[dashed] (3)--(4);
\draw[dashed] (1)--(1.5,1);
\draw[dashed] (1)--(-1,-1.5);
\draw[dashed] (4)--(4.5,1);
\draw[dashed] (4)--(7,-1.5);
\draw[dashed] (2)--(-1,-4.5);
\draw[dashed] (2)--(1.5,-7);
\draw[dashed] (3)--(7,-4.5);
\draw[dashed] (3)--(4.5,-7);

\node at (3,-3) {$v$};
\node at (-2.5,-3) {$H_{v,loc}($};
\node at (8,-3) {$)$};
\end{tikzpicture}
\begin{tikzpicture}[scale=0.9]

\draw (-1,0)--(7,0);
\draw (-1,-3)--(7,-3);
\draw (-1,-6)--(7,-6);
\draw (0,1)--(0,-7);
\draw (3,1)--(3,-7);
\draw (6,1)--(6,-7);

\node[circle,draw,scale=0.7] (1) at (1.5,-1.5) {$x_1$};
\node[circle,draw,scale=0.7] (2) at (1.5,-4.5) {$x_2$};
\node[circle,draw,scale=0.7] (3) at (4.5,-4.5) {$x_3$};
\node[circle,draw,scale=0.7] (4) at (4.5,-1.5) {$x_4$};
\node[circle,draw,scale=0.5] (11) at (2.5,-1.5) {$y_1$};
\node[circle,draw,scale=0.5] (12) at (3.5,-1.5) {$y_1$};
\node[circle,draw,scale=0.5] (21) at (1.5,-2.5) {$y_2$};
\node[circle,draw,scale=0.5] (22) at (1.5,-3.5) {$y_2$};
\node[circle,draw,scale=0.5] (31) at (2.5,-4.5) {$y_3$};
\node[circle,draw,scale=0.5] (32) at (3.5,-4.5) {$y_3$};
\node[circle,draw,scale=0.5] (41) at (4.5,-2.5) {$y_4$};
\node[circle,draw,scale=0.5] (42) at (4.5,-3.5) {$y_4$};

\node[above] at (2,-1.5) {$X_1$};
\node[above] at (4,-1.5) {$X_1$};
\node[above] at (3,-1.5) {$X_1^{\prime}$};
\node[left] at (1.5,-2) {$X_2$};
\node[left] at (1.5,-4) {$X_2$};
\node[left] at (1.5,-3) {$X_2^{\prime}$};
\node[below] at (2,-4.5) {$X_3$};
\node[below] at (4,-4.5) {$X_3$};
\node[below] at (3,-4.5) {$X_3^{\prime}$};
\node[right] at (4.5,-2) {$X_4$};
\node[right] at (4.5,-4) {$X_4$};
\node[right] at (4.5,-3) {$X_4^{\prime}$};
\node at (2.2,-2.2) {$Y$};
\node at (3.8,-2.2) {$Y$};
\node at (2.2,-3.8) {$Y$};
\node at (3.8,-3.8) {$Y$};

\draw[red] (1)--(11);
\draw[red] (1)--(21);
\draw[red] (2)--(22);
\draw[red] (2)--(31);
\draw[red] (3)--(32);
\draw[red] (3)--(42);
\draw[red] (4)--(12);
\draw[red] (4)--(41);
\draw[dashed] (11)--(12);
\draw[dashed] (21)--(22);
\draw[dashed] (31)--(32);
\draw[dashed] (41)--(42);
\draw[red] (11)--(21);
\draw[red] (22)--(31);
\draw[red] (32)--(42);
\draw[red] (12)--(41);

\draw[dashed] (1)--(1.5,1);
\draw[dashed] (1)--(-1,-1.5);
\draw[dashed] (4)--(4.5,1);
\draw[dashed] (4)--(7,-1.5);
\draw[dashed] (2)--(-1,-4.5);
\draw[dashed] (2)--(1.5,-7);
\draw[dashed] (3)--(7,-4.5);
\draw[dashed] (3)--(4.5,-7);

\fill[blue, fill opacity=0.2](1) circle (1.3);
\fill[blue, fill opacity=0.2](2) circle (1.3);
\fill[blue, fill opacity=0.2](3) circle (1.3);
\fill[blue, fill opacity=0.2](4) circle (1.3);

\node at (3,-3) {$v$};
\node at (-3.5,-3) {$=\sum_{Y,X_i^{\prime},y_i}\frac{d_Y\prod_i\sqrt{d_{X_i}d_{X_i^{\prime}}}}{D^2}$};

\end{tikzpicture}
    \caption{$H_{v,loc}$}
    \label{fig1}
\end{figure}
\begin{mythm}
The Levin-Wen Hamiltonian schema is topological.
\end{mythm}

Below we prove the Levin-Wen model satisfies the TQO1 and TQO2 conditions. 

\subsection{TQO1}

To prove the TQO1, let $\Delta_s$ be a subcomplex formed by 2-cells and $\Delta_{\bar{c}}$ be the closure of the complement of $\Delta_s$, and define the following vector spaces.
\begin{align*}
&L_*=\bigotimes_{e\in\Delta_*}\mathbb{C}^{\pi_{\mathcal{C}}}\otimes\bigotimes_{f\in\Delta_*}\text{Hom}(1,(\bigoplus_{X\in\pi_{\mathcal{C}}}X)^{\otimes n(f)})
\\
&L_c=\bigotimes_{e\notin\Delta_s}\mathbb{C}^{\pi_{\mathcal{C}}}\otimes\bigotimes_{f\notin\Delta_s}\text{Hom}(1,(\bigoplus_{X\in\pi_{\mathcal{C}}}X)^{\otimes n(f)})
\\
&L_v=\prod_{f\in\Delta}H_f(L)
\\
&L_{*,v}=\prod_{f\in\Delta_*}H_{f}(L_*)
\end{align*}
where $*=s,\bar{c}$. More explicitly, we have
\begin{align*}
&L_v=\bigoplus_{l}\bigotimes_{f\in\Delta}\text{Hom}(1,X_1^{l,f}\otimes\cdots\otimes X_{n(f)}^{l,f})
\\
&L_{*,v}=\bigoplus_{l_*}\bigotimes_{f\in\Delta_*}\text{Hom}(1,X_1^{l_*,f}\otimes\cdots\otimes X_{n(f)}^{l_*,f})
\end{align*}
where $l,l_*$ are the labeling for edges of $\Delta,\Delta_*$ respectively, and $X_i^{l,f}$ are the labels of edges on $f$.

For any $h\in End(L_s)$, we define $h_g\in End(V)$ as follows.
$$V\stackrel{i_p}{\longrightarrow}L_v\stackrel{i_v}{\longrightarrow}L_s\otimes L_c\stackrel{h\otimes\text{id}}{\longrightarrow}L_s\otimes L_c\stackrel{\prod_{f\in\Delta}H_f}{\longrightarrow}L_v\stackrel{\mathcal{H}_p}{\longrightarrow}V
$$
where $i_p,i_v$ are inclusions.

\begin{mypro}
When $\Delta_s$ is in the interior of subcomplex $\Delta_b$ and $\Delta_b$ is homeomorphic to a disk, $h_g=c\cdot id$ for some $c\in\mathbb{C}$.
\end{mypro}
\begin{proof} Let $W$ be the subspace of $L_{s,v}\otimes L_{\bar{c},v}$ spanned by the basis whose labels on the boundary of $\Delta_s$ and the boundary of $\Delta_{\bar{c}}$ are different. $L_v$ is isomorphic to the subspace of $L_{s,v}\otimes L_{\bar{c},v}$ spanned by the basis whose labels on the boundary of $\Delta_s$ and the bounary $\Delta_{\bar{c}}$ are the same. We have a natural decomposition
$$L_{s,v}\otimes L_{\bar{c},v}=L_v\oplus W
$$
with inner product such that $L_v,W$ are orthogonal to each other. For any $h\in End(L_s)$, define $h_v\in End (L_{s,v})$ by
$$L_{s,v}\stackrel{i_{s,v}}{\longrightarrow}L_s\stackrel{h}{\longrightarrow}L_s\stackrel{\prod_{f\in\Delta_s}H_f}{\longrightarrow}L_{s,v}
$$
where $i_{s,v}$ is the inclusion for $\prod_{f\in\Delta_s}H_f$.

We construct the following diagram
$$\xymatrix{
&L_s\otimes L_c\ar[r]^{h\otimes id}&L_s\otimes L_c\ar[dr]^{\prod_{f\in\Delta}H_f}& 
\\
L_v\ar[ur]^{i_v}\ar[dr]^i& & &L_v
\\
&L_{s,v}\otimes L_{\bar{c},v}\ar[r]^{h_v\otimes id}&L_{s,v}\otimes L_{\bar{c},v}\ar[ur]^p&
}
$$
where $i,p$ are inclusion and projection induced from the above decomposition. We show this diagram is commuting in \textbf{Lemma 1.}
\\
Next the problem is reduced to one of TQFTs. It suffices to prove the \textbf{Lemma 2.}
\end{proof}
\begin{mylem}
The above diagram is commuting.
\end{mylem}
\begin{proof}
For any $\bigotimes_{e\in\Delta} X_{e,0}\otimes\bigotimes_{f\in\Delta}(0,...,x_{j,0}^f,...,0)\in L_{v}$, we have
\begin{align*}
&UP(\bigotimes_{e\in\Delta}X_{e,0}\otimes\bigotimes_{f\in\Delta}(0,...,x_{j,0}^f,...,0))
\\
&=UP(\bigotimes_{e\in\Delta_s}X_{e,0}\otimes\bigotimes_{f\in\Delta_s}(0,...,x_{j,0}^f,...,0)\otimes\bigotimes_{e\notin\Delta_s}X_{e,0}\otimes\bigotimes_{f\notin\Delta_s}(0,...,x_{j,0}^f,...,0))
\\
&=\prod_{f\in\Delta}H_f(\sum_{(1)}a_0^{X_e,x_k^f}\bigotimes_{e\in\Delta_s}X_e\otimes\bigotimes_{f\in\Delta_s}(x_1^f,...,x_{|L|^{n(f)}}^f)\otimes\bigotimes_{e\notin\Delta_s}X_{e,0}\otimes\bigotimes_{f\notin\Delta_s}(0,...,x_{j,0}^f,...,0))
\\
&=\prod_{f\notin\Delta_s}H_f(\sum_{(1)}a_0^{X_e,x_k^f}\bigotimes_{e\in\Delta_s}X_e\otimes\bigotimes_{f\in\Delta_s}(0,...,x_l^f,...,0)\otimes\bigotimes_{e\notin\Delta_s}X_{e,0}\otimes\bigotimes_{f\notin\Delta_s}(0,...,x_{j,0}^f,...,0))
\\
&=\sum_{(2)}a_0^{X_e,x_k^f}\bigotimes_{e\in\Delta_s}X_e\otimes\bigotimes_{f\in\Delta_s}(0,...,x_l^f,...,0)\otimes\bigotimes_{e\notin\Delta_s}X_{e,0}\otimes\bigotimes_{f\notin\Delta_s}(0,...,x_{j,0}^f,...,0)
\end{align*}
where $(1)$ is over all the elements of the basis for $L_s$, $(2)$ is over the elements whose labels on the boundary of $\Delta_s$ are the same as $\bigotimes_{e\in\partial\Delta_s}X_{e,0}$, $x_l^f$ are for $\bigotimes_{e\in\Delta_s}X_e$ and $a_0^{X_e,x_k^f}$ are the coefficients for $\bigotimes_{e\Delta_s}X_{e,0}\otimes\bigotimes_{f\Delta_s}(0,...,x_{j,0}^f,...,0)$ and $\bigotimes_{e\in\Delta_s}X_e\otimes\bigotimes_{f\in\Delta_s}(x_1^f,...,x_{|L|^{n(f)}}^f)$ according to $h$.
\begin{align*}
&DOWN(\bigotimes_{e\in\Delta}X_{e,0}\otimes\bigotimes_{f\in\Delta}(0,...,x_{j,0}^f,...,0))
\\
&=p\cdot(h_v\otimes id)(\bigotimes_{e\in\Delta_s}X_{e,0}\otimes\bigotimes_{f\in\Delta_s}(0,...,x_{j,0}^f,...,0)\otimes\bigotimes_{e^{\prime}\in\Delta_{\bar{c}}}X_{e^{\prime},0}\otimes\bigotimes_{f\in\Delta_{\bar{c}}}(0,...,x_{j,0}^f,...,0))
\\
&=p(\prod_{f\in\Delta_s}H_f(\sum_{(1)}a_0^{X_e,x_k^f}\bigotimes_{e\in\Delta_s}X_e\otimes\bigotimes_{f\in\Delta_s}(x_1^f,...,x_{|L|^{n(f)}}^f))\otimes\bigotimes_{e^{\prime}\in\Delta_{\bar{c}}}X_{e^{\prime},0}\otimes\bigotimes_{f\in\Delta_{\bar{c}}}(0,...,x_{j,0}^f,...,0))
\\
&=p(\sum_{(1)}a_0^{X_e,x_k^f}\bigotimes_{e\in\Delta_s}X_e\otimes\bigotimes_{f\in\Delta_s}(0,...,x_l^f,...,0)\otimes\bigotimes_{e^{\prime}\in\Delta_{\bar{c}}}X_{e^{\prime},0}\otimes\bigotimes_{f\in\Delta_{\bar{c}}}(0,...,x_{j,0}^f,...,0))
\\
&=\sum_{(2)}a_0^{X_e,x_k^f}\bigotimes_{e\in\Delta_s}X_e\otimes\bigotimes_{f\in\Delta_s}(0,...,x_l^f,...,0)\otimes\bigotimes_{e^{\prime}\in\Delta_{\bar{c}}}X_{e^{\prime},0}\otimes\bigotimes_{f\in\Delta_{\bar{c}}}(0,...,x_{j,0}^f,...,0)
\\
&=\sum_{(2)}a_0^{X_e,x_k^f}\bigotimes_{e\in\Delta_s}X_e\otimes\bigotimes_{f\in\Delta_s}(0,...,x_l^f,...,0)\otimes\bigotimes_{e\notin\Delta_s}X_{e,0}\otimes\bigotimes_{f\notin\Delta_s}(0,...,x_{j,0}^f,...,0)
\end{align*}
where $X_{e,0}=X_{e^{\prime},0}$ for any $e=e^{\prime}\in\partial\Delta_s=\partial\Delta_{\bar{c}}$.
\end{proof}
\begin{mylem}
For any $h_v\in End(L_{s,v})$, we can define $h_g\in End(v)$ by
$$V\stackrel{i_p}{\longrightarrow}L_v\stackrel{i}{\longrightarrow}L_{s,v}\otimes L_{\bar{c},v}\stackrel{h_v\otimes\text{id}}{\longrightarrow}L_{s,v}\otimes L_{\bar{c},v}\stackrel{p}{\longrightarrow}L_v\stackrel{\mathcal{H}_p}{\longrightarrow}V
$$
When $\Delta_s$ is in the interior of subcomplex $\Delta_b$ and $\Delta_b$ is homeomorphic to a disk, $h_g=c\cdot id$ for some $c\in\mathbb{C}$.
\end{mylem}
\begin{proof}Without loss of generality, we assume that $\Delta_s$ is the maximal disk formed by 2-cells in the interior of $\Delta_b$. Consider the 2-cells $f$ of $\Delta_b$ on the boundary which means that there exists a vertex of $f$ on the boundary of $\Delta_b$. Since we choose a fine cell complex, they form an annulus $A$ around $\Delta_s$. For any edge $e$ whose interior is in the interior of A, two vertices of $e$ have to be on two boundaries of $A$ respectively. For any $e$ as described above, we construct new $\Delta_{re}$ from $\Delta$ by removing $e$ and denote closure of complement of $\Delta_s$ by $\Delta_{\bar{c},re}$. Define the following vector spaces.
\begin{align*}
&L_{re,v}=\bigoplus_{l_{re}}\bigotimes_{f\in\Delta_{re}}\text{Hom}(1,X_1^{f,l_{re}}\otimes\cdots\otimes X_{n(f)}^{f,l_{re}})
\\
&L_{\bar{c},re,v}=\bigoplus_{l_{\bar{c},re}}\bigotimes_{f\in\Delta_{\bar{c},re}}\text{Hom}(1,X_1^{f,l_{\bar{c},re}}\otimes\cdots\otimes X_{n(f)}^{f,l_{\bar{c},re}})
\end{align*}

We have a natural decomposition
$$L_{s,v}\otimes L_{\bar{c},re,v}=L_{re,v}\oplus W_{re}
$$
where $W_{re}$ is defined as $W$ above. The decomposition induces inclusion $j_m$, projection $p_m$.

Define the isomorphism $re:L_{v}\longrightarrow L_{re,v}$ and $re:L_{\bar{c},v}\longrightarrow L_{\bar{c},re,v}$ by
\begin{align*}
&\bigoplus_{l_*}\bigotimes_{f\in\Delta}\text{Hom}(1,X_1^{f,l_*}\otimes\cdots\otimes X_{n(f)}^{f,l_*})
\\
&=\bigoplus_{l_*\backslash e}\bigotimes_{f\in\Delta\backslash f_1,f_2}\text{Hom}(1,X_1^{f,l_*\backslash e}\otimes\cdots\otimes X_{n(f)}^{f,l_*\backslash e})\otimes\bigoplus_{X\in L}H(f_1,l_*\backslash e,X)\otimes H(f_2,l_*\backslash e,X)
\\
&=\bigoplus_{l_*\backslash e}\bigotimes_{f\in\Delta\backslash f_1,f_2}\text{Hom}(1,X_1^{f,l_*\backslash e}\otimes\cdots\otimes X_{n(f)}^{f,l_*\backslash e})\otimes H(f_1\cup_{e}f_2,l_*\backslash e)
\\
&
\end{align*}
where $H(f_i,l_*\backslash e, X)$ corresponds to the colorings whose restriction on $e$ is $X$, the second equation comes from composition with a twist factor $d_{X}$ and $*=\ \ ,\bar{c}$. According to \cite{Kirillov11}, we have
$$re\cdot\mathcal{H}_p=\mathcal{H}_{p,re}\cdot re
$$
where $\mathcal{H}_{p,re}$ acts on $L_{re,v}$.

After processing the above remove operation, we transform $\Delta$ to $\Delta_{re}$ such that there is only one $e$ whose interior is in the interior of $A$ as shown in \ref{fig2}.
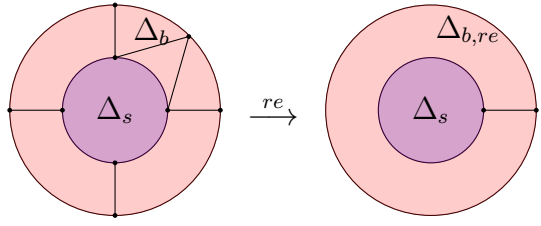
\begin{figure}
    \centering
\begin{tikzpicture}[scale=0.7]
\draw (0,0) circle (1);
\draw (0,0) circle (2);
\filldraw[blue,opacity=0.2] (0,0) circle (1);
\filldraw[red,opacity=0.2] (0,0) circle (2);
\node at (3,0) {$\stackrel{re}{\longrightarrow}$};
\draw (6,0) circle (1);
\draw (6,0) circle (2);
\filldraw[blue,opacity=0.2] (6,0) circle (1);
\filldraw[red,opacity=0.2] (6,0) circle (2);
\filldraw (1,0) circle(1pt);
\filldraw (2,0) circle(1pt);
\filldraw (-1,0) circle(1pt);
\filldraw (-2,0) circle(1pt);
\filldraw (0,1) circle(1pt);
\filldraw (0,2) circle(1pt);
\filldraw (0,-1) circle(1pt);
\filldraw (0,-2) circle(1pt);
\filldraw (1.4,1.4) circle(1pt);
\filldraw (7,0) circle(1pt);
\filldraw (8,0) circle(1pt);

\draw (1,0)--(2,0);
\draw (-1,0)--(-2,0);
\draw (0,1)--(0,2);
\draw (0,-1)--(0,-2);
\draw (0,1)--(1.4,1.4);
\draw (1,0)--(1.4,1.4);
\draw (7,0)--(8,0);

\node at (0,0) {$\Delta_s$};
\node at (0.7,1.5) {$\Delta_b$};
\node at (6,0) {$\Delta_s$};
\node at (6.7,1.5) {$\Delta_{b,re}$};

\end{tikzpicture}

    \caption{$\Delta_{re}$}
    \label{fig2}
\end{figure}

We construct $\Delta_{cut}$ by cutting $A$ from $\Delta_{re}$ and gluing back disks $D_{\alpha},D_{\beta}$ along the boundaries, and $\Delta_{\bar{c},cut}$ be the closure of complement of $\Delta_s$. For any $Z\in Z(\mathcal{C})$, define the following vector spaces.
\begin{align*}
&L_{cut,v,Z}=\bigoplus_{l_{cut,Z}}\bigotimes_{f\in\Delta_{cut}}H(f,l_{cut})\\
&L_{\bar{c},cut,v,Z}=\bigoplus_{l_{\bar{c},cut,Z}}\bigotimes_{f\in\Delta_{\bar{c},cut}}H(f,l_{\bar{c},cut})
\end{align*}

We have a natural decomposition.
$$L_{s,v}\otimes L_{\bar{c},cut,v,Z}=L_{cut,v,Z}\oplus W_{cut,Z}
$$
where $W_{cut,Z}$ is defined as $W$ above. The decomposition induces inclusion $j_{d,Z}$, projection $p_{d,Z}$.

According to construction in \cite{KB10}, there is isomorphism $G:\bigoplus_{Z\in Z(\mathcal{C})}L_{cut,v,Z}\longrightarrow L_{re,v}$ and $G:\bigoplus_{Z\in Z(\mathcal{C})}L_{\bar{c},cut,v,Z}\longrightarrow L_{\bar{c},re,v}$ satisfying
$$G\cdot\mathcal{H}_{p,cut,Z}=\mathcal{H}_{p,re}\cdot G
$$

We denote $L_{s,v}\otimes L_{\bar{c},re,v}$ by $L_m$, $\bigoplus_{Z\in Z(\mathcal{C})}L_{s,v}\otimes L_{\bar{c},cut,v,Z}$ by $L_d$. Then we construct the following commuting diagram.
$$\xymatrix{
V\ar[r]^{i_p}\ar[d]^{re}
&L_v\ar[d]^{re}\ar[r]^j
&L_{s,v}\otimes L_{\bar{c},v}\ar[r]^{h_v\otimes id}\ar[d]^{id\otimes re}
&L_{s,v}\otimes L_{\bar{c},v}\ar[r]^p\ar[d]^{id\otimes re}
&L_v\ar[r]^{\mathcal{H}_p}\ar[d]^{re}
&V\ar[d]^{re}
\\
V_{re}\ar[r]^{i_{re,p}}\ar[d]^{G^{-1}}
&L_{re,v}\ar[r]^{j_m}\ar[d]^{G^{-1}}
&L_m\ar[r]^{h_v\otimes id}\ar[d]^{id\otimes G^{-1}}
&L_m\ar[r]^{p_m}\ar[d]^{id\otimes G^{-1}}
&L_{re,v}\ar[r]^{\mathcal{H}_{re,p}}\ar[d]^{G^{-1}}
&V_{re}\ar[d]^{G^{-1}}
\\
\bigoplus V_{cut,Z}\ar[r]^{i_{cut,p}}
&\bigoplus L_{cut,v,Z}\ar[r]^{j_d}
&L_d\ar[r]^{h_v\otimes id}
&L_d\ar[r]^{p_d}
&\bigoplus L_{cut,v,Z}\ar[r]^{\mathcal{H}_{cut,p}}
&\bigoplus V_{cut,Z}
}
$$
where sum is over $Z\in Z(\mathcal{C})$.

According to \cite{KB10}, $\bigoplus_{Z\in\mathcal{C}}V_{cut,Z}=V_{cut,1}$. Since $\Delta_{cut}$ is the disjoint union of $\Delta_{\bar{c},cut}\cup D_{\beta}$ and $\Delta_s\cup D_{\alpha}$, $L_{cut,v,1}=L_{\Delta_{\bar{c},cut}\cup D_{\beta},1}\otimes L_{\Delta_s\cup D_{\alpha},1}$ and $\mathcal{H}_{cut,p,1}=\mathcal{H}_{\Delta_{\bar{c},cut}\cup D_{\beta},1}\otimes\mathcal{H}_{\Delta_s\cup D_{\alpha},1}$. Since $h_v\in End(L_{s,v})$, $p_{d,1}\cdot(h_v\otimes id)_1\cdot j_d=h_{S^2}\otimes id\in End(L_{cut,1})$. Then the composition of bottom row is $h_{S^2}\otimes id\in End(V_{cut,1})=End(V_{\Delta_s\cup D_{\alpha},1}\otimes V_{\Delta_{\bar{c},cut}\cup D_{\beta},1})$. Since $\Delta_s\cup D_{\alpha}$ is a sphere, dim$V_{\Delta_s\cup D_{\alpha},1}=1$, $h_{S^2}=c$ for some $c\in\mathcal{C}$. Then $h_{S^2}\otimes id=c\cdot id$.
\end{proof}

\subsection{TQO2}
Now we prove that the Levin-Wen model satisfies TQO2, i.e., for any $h:L_s\longrightarrow L_s$, $h\otimes id\prod_{f\in\Delta}H_f\prod_{v\in\Delta}H_v=0$ implies 
$$h\otimes id\prod_{f\in\Delta_b}H_f\prod_{v\in\Delta_s}H_v=0
$$.

Let $V_b=\prod_{v\in\Delta_s}H_v\prod_{f\in\Delta_b}H_f(L)$. We consider the following diagram.
$$\xymatrix{
V\ar[r]^{i}\ar[d]^{re}
&V_b\ar[r]^{j}\ar[d]^{re}
&L_s\otimes L_c\ar[r]^{h\otimes id}
&L_s\otimes L_c
\\
V_{re}\ar[r]^{i_{re}}\ar[d]^{G^{-1}}
&V_{re,b}\ar[d]^{G^{-1}}
& & 
\\
\bigoplus_Z V_{cut,Z}\ar[r]^{i_{cut}}
&\bigoplus_Z V_{cut,Z,b}
& & }
$$
where $i,j,i_{re},i_{cut}$ denote the inclusions corresponding to $\prod_{f\notin\Delta_b}H_f\prod_{v\notin\Delta_s}H_s$, $\prod_{f\in\Delta_b}H_f\prod_{v\in\Delta_s}H_s$, $\prod_{f\notin\Delta_{b,re}}H_f\prod_{v\notin\Delta_s}H_s$, $\prod_{f\in\Delta_{\bar{c},cut}\cup D_{\beta}}H_f\prod_{v\in\Delta_{\bar{c},cut}\cup D_{\beta}}H_v$ respectively. Then all we need is to show $h\otimes id \circ j\circ i=0$ implies $h\otimes id\circ j=0$.

First we show that two rectangles in the diagram are commuting. Compare with the argument in TQO1, all we need is to replace the TQFT operators by local operators. Since $H_f$ are the same in $\Delta\backslash\Delta_b$, $\Delta_{re}\backslash\Delta_{b,re}$, $\Delta_{\bar{c},cut}$, we just consider $H_v$. Actually we have the following two figures Figs. $3$ and $4$. 
The blue lines represent the edges removed.

\begin{figure}[!ht]\label{TQO2-1}
    \centering
\begin{tikzpicture}
\draw (-0.5,0)--(6.5,0);
\draw (0,0.5)--(0,-3.5);
\draw (-0.5,-3)--(6.5,-3);
\draw (6,0.5)--(6,-3.5);
\draw[blue] (3,0)--(3,-3);

\node(1)[circle,draw,scale=0.7] at (1.5,-1.5) {$x_1$};
\node(2)[circle,draw,scale=0.7] at (4.5,-1.5) {$x_2$};
\node(3)[circle,draw,scale=0.5] at (1.5,-0.5) {$y_1$};
\node(4)[circle,draw,scale=0.5] at (2.5,-1.5) {$y$};
\node(5)[circle,draw,scale=0.5] at (3.5,-1.5) {$y$};
\node(6)[circle,draw,scale=0.5] at (4.5,-0.5) {$y_2$};

\draw[red] (1)--(3);
\draw[red] (1)--(4);
\draw[red] (4)--(5);
\draw[red] (5)--(2);
\draw[red] (2)--(6);

\draw[dashed] (3)--(1.5,0.5);
\draw[dashed] (1)--(-0.5,-1.5);
\draw[dashed] (1)--(1.5,-3.5);
\draw[dashed] (6)--(4.5,0.5);
\draw[dashed] (2)--(4.5,-3.5);
\draw[dashed] (2)--(6.5,-1.5);
\draw[red] (3)--(4);
\draw[red] (5)--(6);

\node[left] at (1.5,-1) {$X_1$};
\node[left] at (1.5,0) {$X_1^{\prime}$};
\node[below] at (2,-1.5) {$X$};
\node[below] at (4,-1.5) {$X$};
\node[below] at (3,-1.5) {$X^{\prime}$};
\node[right] at (4.5,-1) {$X_2$};
\node[right] at (4.5,0) {$X_2^{\prime}$};
\node at (2.2,-1) {$Y$};
\node at (3.8,-1) {$Y$};
\fill[blue,fill opacity=0.2] (0.3,-0.2) rectangle (5.7,-2.8);

\node at (-4,-1.5) {$\sum_{Y,X^{\prime},X_i^{\prime},y,y_i}\frac{d_Y\sqrt{d_{X^{\prime}}}\sqrt{d_{X}d_{X^{\prime}}}\prod_i\sqrt{d_{X_i}d_{X_i^{\prime}}}}{D^2}$};

\end{tikzpicture}
\begin{tikzpicture}
\draw (-0.5,0)--(6.5,0);
\draw (0,0.5)--(0,-3.5);
\draw (-0.5,-3)--(6.5,-3);
\draw (6,0.5)--(6,-3.5);
\draw[blue] (3,0)--(3,-3);

\node(1)[circle,draw,scale=0.7] at (1.5,-1.5) {$x_1$};
\node(2)[circle,draw,scale=0.7] at (4.5,-1.5) {$x_2$};
\node(3)[circle,draw,scale=0.5] at (1.5,-0.5) {$y_1$};
\node(6)[circle,draw,scale=0.5] at (4.5,-0.5) {$y_2$};

\draw[red] (1)--(3);
\draw[red] (1)--(2);
\draw[red] (2)--(6);
\draw[red] (3)--(6);

\draw[dashed] (3)--(1.5,0.5);
\draw[dashed] (1)--(-0.5,-1.5);
\draw[dashed] (1)--(1.5,-3.5);
\draw[dashed] (6)--(4.5,0.5);
\draw[dashed] (2)--(4.5,-3.5);
\draw[dashed] (2)--(6.5,-1.5);

\node[left] at (1.5,-1) {$X_1$};
\node[left] at (1.5,0) {$X_1^{\prime}$};
\node[below] at (3,-1.5) {$X$};
\node[right] at (4.5,-1) {$X_2$};
\node[right] at (4.5,0) {$X_2^{\prime}$};
\node at (3,-0.5) {$Y$};

\fill[blue,fill opacity=0.2] (0.3,-0.2) rectangle (5.7,-2.8);

\node at (-3,-1.5) {$=\sum_{Y,X_i^{\prime},y_i}\frac{d_Y\sqrt{d_X}\prod_i\sqrt{d_{X_i}d_{X_i^{\prime}}}}{D^2}$};

\end{tikzpicture}

\begin{tikzpicture}
\draw (-0.5,0)--(6.5,0);
\draw (0,0.5)--(0,-3.5);
\draw (-0.5,-3)--(6.5,-3);
\draw (6,0.5)--(6,-3.5);
\draw[blue] (3,0)--(3,-3);

\node(1)[circle,draw,scale=0.7] at (3,-1.5) {$x_1\circ_{X}x_2$};
\node(3)[circle,draw,scale=0.5] at (1.5,-0.5) {$y_1$};
\node(6)[circle,draw,scale=0.5] at (4.5,-0.5) {$y_2$};

\draw[red] (1)--(3);
\draw[red] (3)--(6);
\draw[dashed] (3)--(1.5,0.5);
\draw[dashed] (1)--(-0.5,-1.5);
\draw[dashed] (1)--(1.5,-3.5);
\draw[dashed] (6)--(4.5,0.5);
\draw[dashed] (1)--(4.5,-3.5);
\draw[dashed] (1)--(6.5,-1.5);
\draw[red] (1)--(6);

\node[left] at (2,-1) {$X_1$};
\node[left] at (1.5,0) {$X_1^{\prime}$};
\node[right] at (4,-1) {$X_2$};
\node[right] at (4.5,0) {$X_2^{\prime}$};
\node at (3,-0.5) {$Y$};

\fill[blue,fill opacity=0.2] (0.3,-0.2) rectangle (5.7,-2.8);

\node at (-3,-1.5) {$=\sum_{Y,X_i^{\prime},y_i}\frac{d_Y\sqrt{d_X}\prod_i\sqrt{d_{X_i}d_{X_i^{\prime}}}}{D^2}$};

\end{tikzpicture}
    \caption{$re\circ H_{v,loc}=H_{v,re,loc}\circ re$}
    \label{fig3}
\end{figure}

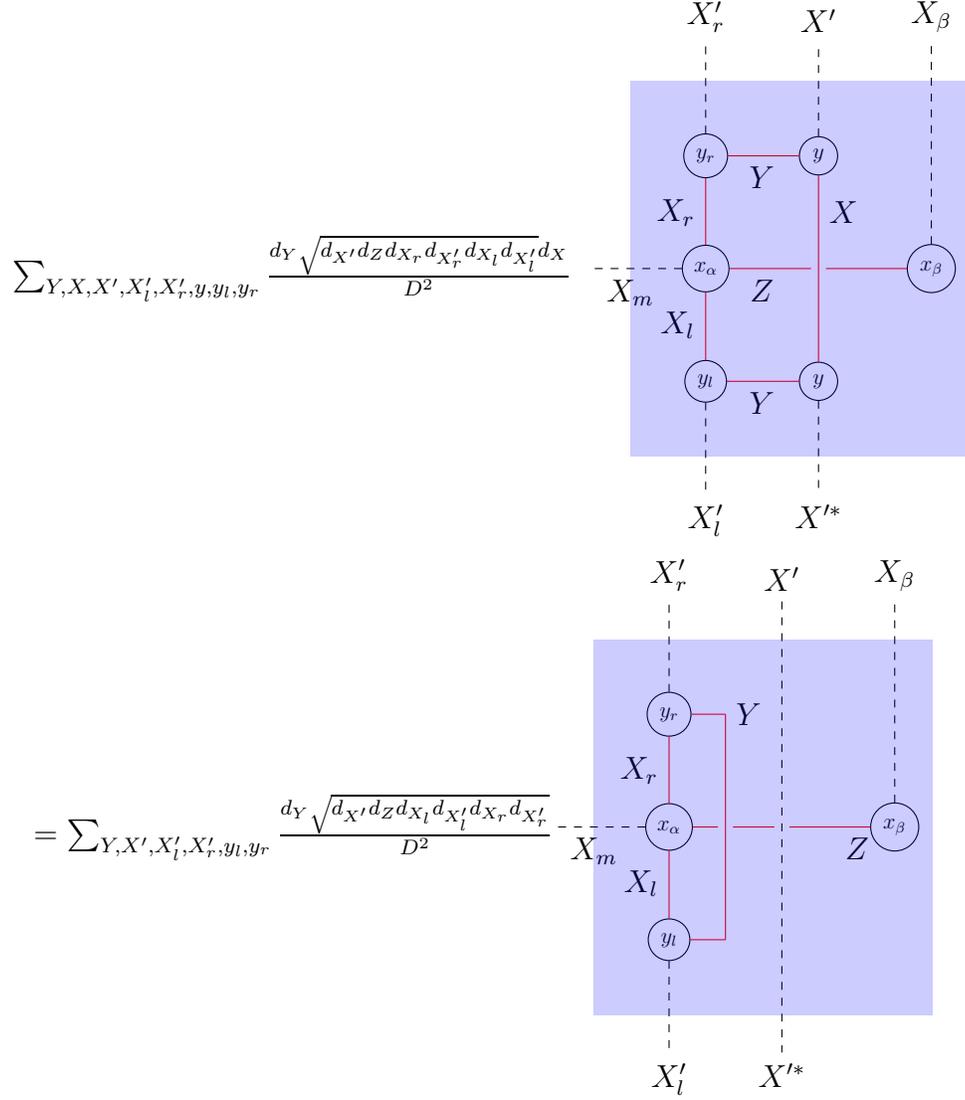
\begin{figure}[!ht]\label{TQO2-2}
    \centering
\begin{tikzpicture}
\node[circle,draw,scale=0.7](1) at (0,0) {$y_r$};
\node[circle,draw,scale=0.7](2) at (1.5,0) {$y$};
\node[circle,draw,scale=0.7](3) at (0,-1.5) {$x_{\alpha}$};
\node[circle,draw,scale=0.7](4) at (0,-3) {$y_l$};
\node[circle,draw,scale=0.7](5) at (1.5,-3) {$y$};
\node[circle,draw,scale=0.7](6) at (3,-1.5) {$x_{\beta}$};

\draw[dashed] (1)--(0,1.5);
\draw[red] (1)--(3);
\draw[red] (1)--(2);
\draw[red] (3)--(4);
\draw[red] (4)--(5);
\draw[red] (2)--(5);
\draw[dashed] (4)--(0,-4.5);
\draw[dashed] (5)--(1.5,-4.5);
\draw[dashed] (2)--(1.5,1.5);
\draw[dashed] (6)--(3,1.5);
\draw[dashed] (3)--(-1.5,-1.5);
\draw[red] (3)--(1.4,-1.5);
\draw[red] (1.6,-1.5)--(6);

\node[above] at (0,1.5) {$X_{r}^{\prime}$};
\node[above] at (1.5,1.5) {$X^{\prime}$};
\node[above] at (3,1.5) {$X_{\beta}$};
\node[below] at (0,-4.5) {$X_{l}^{\prime}$};
\node[below] at (1.5,-4.5) {$X^{\prime*}$};
\node[left] at (0,-0.75) {$X_{r}$};
\node[left] at (0,-2.25) {$X_{l}$};
\node[below] at (-1,-1.5) {$X_{m}$};
\node[right] at (1.5,-0.75) {$X$};
\node[below] at (0.75,0) {$Y$};
\node[below] at (0.75,-3) {$Y$};
\node[below] at (0.75,-1.5) {$Z$};

\fill[blue,fill opacity=0.2] (-1,1) rectangle (3.5,-4);

\node at (-5.5,-1.5) {$\sum_{Y,X,X^{\prime},X_l^{\prime},X_r^{\prime},y,y_l,y_r}\frac{d_Y\sqrt{d_{X^{\prime}}d_Zd_{X_r}d_{X_r^{\prime}}d_{X_l}d_{X_l^{\prime}}}d_X}{D^2}$};

\end{tikzpicture}
\begin{tikzpicture}
\node[circle,draw,scale=0.7](1) at (0,0) {$y_r$};
\node[circle,draw,scale=0.7](3) at (0,-1.5) {$x_{\alpha}$};
\node[circle,draw,scale=0.7](4) at (0,-3) {$y_l$};
\node[circle,draw,scale=0.7](6) at (3,-1.5) {$x_{\beta}$};

\draw[dashed] (1)--(0,1.5);
\draw[red] (1)--(3);
\draw[red] (1)--(0.75,0);
\draw[red] (0.75,0)--(0.75,-3);
\draw[red] (3)--(4);
\draw[red] (4)--(0.75,-3);
\draw[dashed] (4)--(0,-4.5);
\draw[dashed] (6)--(3,1.5);
\draw[dashed] (3)--(-1.5,-1.5);
\draw[dashed] (1.5,1.5)--(1.5,-4.5);
\draw[red] (3)--(0.65,-1.5);
\draw[red] (0.85,-1.5)--(1.4,-1.5);
\draw[red] (1.6,-1.5)--(6);

\node[above] at (0,1.5) {$X_{r}^{\prime}$};
\node[above] at (1.5,1.5) {$X^{\prime}$};
\node[above] at (3,1.5) {$X_{\beta}$};
\node[below] at (0,-4.5) {$X_{l}^{\prime}$};
\node[left] at (0,-0.75) {$X_{r}$};
\node[left] at (0,-2.25) {$X_{l}$};
\node[below] at (-1,-1.5) {$X_{m}$};
\node[right] at (0.75,0) {$Y$};
\node[below] at (2.5,-1.5) {$Z$};
\node[below] at (1.5,-4.5) {$X^{\prime*}$};

\fill[blue,fill opacity=0.2] (-1,1) rectangle (3.5,-4);

\node at (-5,-1.5) {$=\sum_{Y,X^{\prime},X_l^{\prime},X_r^{\prime},y_l,y_r}\frac{d_Y\sqrt{d_{X^{\prime}}d_Zd_{X_l}d_{X_l^{\prime}}d_{X_r}d_{X_r^{\prime}}}}{D^2}$};

\end{tikzpicture}
    \caption{$G\circ H_{v,cut,loc}=H_{v,re,loc}\circ G$}
    \label{fig4}
\end{figure}

Since all the vertical arrows are isomorphic, we just show $h\otimes id\circ j\circ re\circ G^{-1}\circ i_{cut}=0$ implies $h\otimes id\circ j\circ re\circ G^{-1}=0$. Since $\Delta_s\cup D_{\alpha}$ is a sphere, according to TQO1, the bottom row becomes $V_{cut,1}\stackrel{id\otimes i_{out}}{\longrightarrow}V_{cut,1,b}$ where $V_{cut,1}=V_{\Delta_s\cup D_{\alpha},1}\otimes V_{\Delta_{\bar{c},cut}\cup D_{\beta},1}$, $V_{cut,1,b}=V_{\Delta_s\cup D_{\alpha},1}\otimes H_{D_{\beta}}(L_{\Delta_{\bar{c},cut}\cup D_{\beta},1})$. Since $\text{dim}(V_{\Delta_s\cup D_{\alpha},1})=1$, suppose $V_{\Delta_s\cup D_{\alpha},1}$ is spanned by $c\in L_{V_{\Delta_s\cup D_{\alpha},1}}$. According to the construction of $G,re$, for any $c^{\prime}\in V_{\Delta_{\bar{c},cut}\cup D_{\beta},1}$,
$$j\circ re\circ G^{-1}\circ i_{cut}(c\otimes c^{\prime})=c_s\otimes c^{\prime\prime}\in L_s\otimes L_c
$$
where $c_s$ is the restriction of $c$ on $\Delta_s$ and $c^{\prime\prime}$ depends on $c,c^{\prime}$. Since $h\in End(L_s)$, $h\otimes id (c_s\otimes c^{\prime\prime})=h(c_s)\otimes c^{\prime\prime}$. $h\otimes id\circ j\circ re\circ G^{-1}\circ i_{out}=0$ and $V_{\Delta_{\bar{c},cut}\cup D_{\beta},1}\neq0$ imply $h(c_s)=0$. Thus for any $c_{cut}\in H_{D_{\beta}}(L_{\Delta_{\bar{c},cut}\cup D_{\beta},1})$
$$h\otimes id \circ j\circ re\circ G^{-1}(c_s\otimes c_{cut})=h(c_s)\otimes c_{cut}^{\prime}=0
$$
\clearpage
\newpage

\section{Hamiltonian schemas for DW TQFTs}\label{DW}

The formalism in \cite{QW20} is used for DW theory in this section, which is a generalization of \cite{Wakui92}.  We realize DW TQFTs by the following Hamiltonion schemas based on this approach. Our proof of TQO1 and TQO2 follows \cite{Cui19}.

\subsection{The Hamiltonian schemas to realize DW}
Let $M$ be a compact smooth $n$-manifold with triangulation $\Delta$, and $G$ a finite group. We define the coloring $\sigma_{\Delta}$ on $\Delta$ which assigns to each oriented edge an element of $G$ and satisfies compatible conditions. Then we have vector spaces $\tilde{V}(\Delta,G)=\mathbb{C}\{\sigma_{\Delta}\}$ and $V(\Delta,G)$ which is the image of $\tilde{V}(\Delta,G)$ under $M\times I$ coming from DW TQFT. If no confusion arises, we will omit $G$.

Choosing an orientation for each edge of $\Delta$, we define the Hilbert space $L=\bigotimes_{e\in\Delta}\mathbb{C}^G$. For any vertex $v$, edge $e_v$ incident to $v$, $g\in G$ and $g_e\in G$ labeling edge $e$, define the action of $g$ on $g_e$ by $g\cdot g_e=g^{-1}g_e$ if $e$ is from $v$ and $g\cdot g_e=g_eg$ if $e$ is to $v$. Define vertex operator $H_v$ on $L$ by
$$H_v(\bigotimes_{e\in T}g_e)=\frac{1}{|G|}\sum_{g\in G}\bigotimes_{e_v}g\cdot g_{e_v}\otimes\bigotimes_{\text{other}}g_e
$$
For any triangle $f$, choose a cycle orientation $e_1e_2e_3$ on $f$. Define face operator $H_f$ on $L$ by $H_f=id$ if $g_{e_1}g_{e_2}g_{e_3}=1$ and $H_f=0$ otherwise. $H_f$ deos not depend on the choice of orientation we choose above. Then we have $H_v,H_f$ are projective and $H_vH_{v^{\prime}}=H_{v^{\prime}}H_v$, $H_fH_{f^{\prime}}=H_{f^{\prime}}H_f$, $H_vH_f=H_fH_v$. Moreover we have $\tilde{V}(\Delta)=\prod_{f}H_f(L)$ and $V(\Delta)=\prod_{v,f}H_vH_f(L)$.

\begin{mythm}
The state-sum schema for DW TQFT is topological.
\end{mythm}

\subsection{TQO1}
In this section we prove the DW TQFTs have the TQO1 and TQO2 properties. 

To show that DW theory has error correction code property TQO 1, We formulate the problem for DW TQO1 as follows.

For any subcomplex $\Delta_s$ of $\Delta$, we have vector space $\tilde{V}(\Delta_s)$ spanned by colorings $\sigma_{\Delta_s}$ on $\Delta_s$. Let $\Delta_{\bar{c},s}$ be the subcomplex by removing the interior of $\Delta_s$ and $\tilde{V}(\Delta_{\bar{c},s})$ the corresponding vector space. Then we have a canonical decomposition $\tilde{V}(\Delta_s)\otimes\tilde{V}(\Delta_{\bar{c},s})=\tilde{V}(\Delta)\oplus\mathbb{C}\{\sigma_{\Delta_s}\otimes\sigma_{\Delta_{\bar{c},s}}|\sigma_{\Delta_s}\neq\sigma_{\Delta_{\bar{c},s}}\text{ on boundary}\}$. For any linear map $h_s$ on $\tilde{V}(\Delta_s)$, we have a corresponding map $V(h_s)$ on $V(\Delta)$ by 
$$V(\Delta)\stackrel{i}{\longrightarrow}\tilde{V}(\Delta)\stackrel{j}{\longrightarrow}\tilde{V}(\Delta_s)\otimes\tilde{V}(\Delta_{\bar{c},s})\stackrel{h_s\otimes id}{\longrightarrow}\tilde{V}(\Delta_s)\otimes\tilde{V}(\Delta_{\bar{c},s})\stackrel{p}{\longrightarrow}\tilde{V}(\Delta)\stackrel{M\times I}{\longrightarrow}V(\Delta)
$$
where $i,j$ are natural inclusions and $p$ is natural projection.

We have the following proposition.
\begin{mypro}
When $\Delta_s$ is homeomorphic to a $n$-ball in $M$, for any $h_s\in End(\tilde{V}(\Delta_s))$, $V(h_s)$ is a multiple of identity.
\end{mypro}
\begin{proof}
Choose a vertex of $\Delta_s$ as the base point. We give a basis $\{e_{i,j,k}\}$ for $\tilde{V}(\Delta)$ where $i$ is for the conjugate classes of holonomy maps, $j$ is for the holonomy maps in $i$, $k$ is for the colorings in $j$. We endow the standard inner product. According to \cite{QW20}, we give a basis $\{f_i\}$ for $V(\Delta)$ by $f_i=\sum_{j,k}e_{i,j,k}$. Similarly, we give a basis $\{e^{\prime}_{a,b,c}\}$ for $\tilde{V}(\Delta_s)$.
For any $h_s\in End(\tilde{V}(\Delta_s))$, we have matrix $a_{a,b,c}^{r,s,t}$ by $h_s(e_{a,b,c}^{\prime})=\sum_{r,s,t} a_{a,b,c}^{r,s,t}e^{\prime}_{r,s,t}$. For any $i$, we have
\begin{align*}
V(h_s)(f_i)&=V(h_s)(\sum_{j,k}e_{i,j,k})=\sum_{a,b,c}V(h_s)(\sum_{e_{i,j,k}|_{\Delta_s}=e^{\prime}_{a,b,c}}e_{i,j,k})\\
&=\sum_{a,b,c}M\times I\cdot p(\sum_{r,s,t}a^{r,s,t}_{a,b,c}e^{\prime}_{r,s,t}\otimes\sum_{e_{i,j,k}|_{\Delta_s}=e^{\prime}_{a,b,c}}e_{i,j,k}|_{{\Delta_{\bar{c},s}}})\\
&=\sum_{a,b,c}M\times I(\sum_{e^{\prime}_{r,s,t}|_{\partial\Delta_s}=e^{\prime}_{a,b,c}|_{\partial\Delta_s}}a_{a,b,c}^{r,s,t}e^{\prime}_{r,s,t}\otimes\sum_{e_{i,j,k}|_{\Delta_s}=e^{\prime}_{a,b,c}}e_{i,j,k}|_{\Delta_{\bar{c},s}})\\
\end{align*}
Since $e^{\prime}_{r,s,t}|_{\partial\Delta_s}=e_{i,j,k}|_{\partial\Delta_s}=e^{\prime}_{a,b,c}|_{\partial\Delta_S}$, $e^{\prime}_{r,s,t}\otimes e_{i,j,k}|_{\Delta_{\bar{c},s}}$ corresponds to $e_{i^{\prime},j^{\prime},k^{\prime}}$ on $\Delta$. 

To compute the coefficient for $f_i$, we compute the inner product $<V(h_s)(f_i),f_j>$. According to Lemma 5., $<V(h_s)(f_i),f_j>=0$ for any $i\neq j$. According to Lemma 4., Lemma 5., 
\begin{align*}
<V(h_s)(f_i),f_i>&=<M\times I(\sum_{a,b,c}\sum_{e^{\prime}_{r,s,t}|_{\partial\Delta_s}=e^{\prime}_{a,b,c}|_{\partial\Delta_s}}\sum_{e_{i,j,k}|_{\Delta_s}=e^{\prime}_{a,b,c}}a_{a,b,c}^{r,s,t}e_{i,j^{\prime},k^{\prime}}),f_i>
\\
&=\sum_{a,b,c}\sum_{e^{\prime}_{r,s,t}|_{\partial\Delta_s}=e^{\prime}_{a,b,c}|_{\partial\Delta_s}}a_{a,b,c}^{r,s,t}[G:C_{G}(Im(i))]|G|^{|V|-|V^{\prime}|}
\\
&\frac{1}{[G:C_G(Im(i))]|G|^{|V|-1}}<f_i,f_i>
\\
&=\sum_{a,b,c}\sum_{e^{\prime}_{r,s,t}|_{\partial\Delta_s}=e^{\prime}_{a,b,c}|_{\partial\Delta_s}}a_{a,b,c}^{r,s,t}|G|^{1-|V^{\prime}|}<f_i,f_i>
\end{align*}
We have
$$V(h_s)=(\sum_{a,b,c}\sum_{e^{\prime}_{r,s,t}|_{\partial\Delta_s}=e^{\prime}_{a,b,c}|_{\partial\Delta_s}}a_{a,b,c}^{r,s,t}|G|^{1-|V^{\prime}|})id
$$
\end{proof}
\begin{mylem}
For any $e^{\prime}_{a,b,c}$ and $(i,j)$, there exists $e_{i,j,k}$ such that $e_{i,j,k}|_{\Delta_s}=e^{\prime}_{a,b,c}$
\end{mylem}
\begin{proof}
Choose any $e_{i,j,k}$ for $(i,j)$. Let $e^{\prime}_{d,e,f}=e_{i,j,k}|_{\Delta_s}$. Since $\Delta_s$ is a $n$-ball with trivial fundamental group, we have $a=d,b=e$. Then we can use the method in \cite{QW20} to convert $e^{\prime}_{a,b,f}$ to $e^{\prime}_{a,b,c}$ by coloring the vertices of $\Delta_s$ except the base point. Correspondingly we get $e_{i,j,k^{\prime}}$ we need.
\end{proof}
\begin{mylem}
For any $e^{\prime}_{a,b,c}$ and $i$, there exist exactly $[G:C_{G}(Im(i))]|G|^{|V|-|V^{\prime}|}$ $e_{i,j,k}$ such that $e_{i,j,k}|_{\Delta_s}=e^{\prime}_{a,b,c}$ where $C_G$ is the centralizer and $|V|,|V^{\prime}|$ are the numbers of vertices in $\Delta,\Delta_s$.
\end{mylem}
\begin{proof}
According to Lemma 3, for $e^{\prime}_{a,b,c}$ and $i$, there exists at least one $e_{i,j,k}$ such that $e_{i,j,k}|_{\Delta_s}=e^{\prime}_{a,b,c}$. For one fixed $j$ for $i$, all the $k$'s can be obtained by coloring the vertices of $\Delta$ except the ones in $\Delta_s$. Thus there are exactly $|G|^{|V|-|V^{\prime}|}$ $e_{i,j,k}$ for $(i,j)$. There are exactly $[G:C_G(Im(i))]$ $j$ for $i$. We finish.
\end{proof}

\begin{mylem}
$i^{\prime}=i$.
\end{mylem}
\begin{proof}
It suffices to define a coloring on the vertices in the interior of $\Delta_s$, such that we can convert $e_{i,j,k}$ to $e_{i^{\prime},j^{\prime},k^{\prime}}$ by this coloring.

For each vertex $v^{\prime}$ in the interior of $\Delta_s$, choose a path formed by edges in $\Delta_s$ connecting $v^{\prime}$ to one vertex $\partial v$ on $\partial\Delta_s$. Suppose that we always use $1\in G$ to color the vertices on $\partial\Delta_s$. Then there exists a unique $g\in G$ color $v^{\prime}$ such that all the colorings of $e_{i,j,k}$ on the edges of the path can be converted to the corresponding ones of $e_{i^{\prime},j^{\prime},k^{\prime}}$.

Next let us show $e_{i,j,k}$ is changed to $e_{i^{\prime},j^{\prime},k^{\prime}}$ by the above coloring. Since the coloring does not change the colorings on the edges on $\Delta_{\bar{c},s}$, we just consider the ones inside $\Delta_s$, i.e. at least one of their end points in the interior of $\Delta_s$. For the edges on the paths chosen, it is trivial. Let $e$ be a edge which is not on the paths and with two end points $v_1^{\prime},v_2^{\prime}$. There exist two paths constructed above, $A_1,A_2$, connecting them to $\partial v_1,\partial v_2$ on $\partial\Delta_s$. There exists one path $A_3$ on $\partial\Delta_s$ connecting $\partial v_1,\partial v_2$. Now we use the same notation for the paths and elements coloring them and $g_{v}$ to denote the element coloring $v$. Since $\Delta_s$ is a $n$-ball, we have $A_3A_2^{-1}e_{i,j,k}(e)A_1=A_3A_2^{-1}g_{v_2^{\prime}}e_{i^{\prime},j^{\prime},k^{\prime}}(e)g_{v_1^{\prime}}^{-1}A_1$. We have $e_{i^{\prime},j^{\prime},k^{\prime}}(e)=g_{v_2^{\prime}}^{-1}e_{i,j,k}(e)g_{v_1^{\prime}}$
\end{proof}

\subsection{TQO2}
Now we show DW-theory has TQO 2, and formulate DW TQO2 as follows.

Let $\Delta_s$ be a subcomplex of $\Delta$ formed by $n$-cells and $\Delta_b$ be a subcomplex containing $\Delta_s$ as interior. Define $V_b=\prod_{v\in\mathring{\Delta}_b,f\in\Delta_b}H_vH_f(L)$, $L_s=\bigotimes_{e\in\Delta_s}\mathbb{C}^G$ and $L_{c,s}=\bigotimes_{e\notin\Delta_s}\mathbb{C}^G$ which implies that $L=L_s\otimes L_{c,s}$. For any $h_s\in End(L_s)$, we have an extension $h_b\in End(V_b)$ by
$$V_b\stackrel{i_b}{\longrightarrow}L_s\otimes L_c\stackrel{h_s\otimes id}{\longrightarrow}L_s\otimes L_c\stackrel{\prod_{v\in\mathring{\Delta}_b,f\in\Delta_b}H_vH_f}{\longrightarrow}V_b
$$
where $i_b$ is the inclusion for $\prod_{v,f}H_vH_f$.
\begin{mypro}
When $\Delta_b$ is homeomorphic to a $n$-ball, $h_b=c\cdot id$ for some $c\in\mathbb{C}$.
\end{mypro}
\begin{proof}
Without loss of generality, we assume that $\Delta_s$ is homeomorphic to a disk inside $\Delta_b$. We can rewrite $L,V_b$ as follows
$$L=L_b\otimes L_{c,b}
$$
$$V_b=\prod_{v\in\mathring{\Delta}_b,f\in\Delta_b}H_vH_f(L_b)\otimes L_{c,b}
$$
Since $h_s\otimes id,\prod_{v\in\mathring{\Delta_b},f\in\Delta_b}H_vH_f$ are supported in the first factors, all we need is to consider the following truncation diagram
$$V_{b,tr}\stackrel{i_b}{\longrightarrow}L_s\otimes L_{b\backslash s}\stackrel{h_s\otimes id}{\longrightarrow}L_s\otimes L_{b\backslash s}\stackrel{\prod_{v\in\mathring{\Delta}_b,f\in\Delta_b}H_vH_f}{\longrightarrow}V_{b,tr}
$$

Let $V_{*,f}=\prod_{f\in\Delta_*}H_f(L_*)$ where $*=s,b,\overline{b\backslash s}$. We have the decomposition $V_{s,f}\otimes V_{\overline{b\backslash s},f}=V_{b,f}\oplus W$ where $W$ is spanned by the colorings on $\Delta_s,\Delta_{\overline{b\backslash s}}$ whose restriction on the boundary are different. For any $h_s\in End(L_s)$, we define $h_{s,f}\in End(V_{b,f})$ by 
$$V_{b,f}\stackrel{i_{b,f}}{\longrightarrow}L_s\stackrel{h_s}{\longrightarrow}L_s\stackrel{\prod_{f\in\Delta_s}H_f}{\longrightarrow}V_{b,f}
$$
Let $V_{s,f}=\prod_{f\in\Delta_b}H_f(L_b)$. Then we have the following commuting diagram
$$\xymatrix{
&L_s\otimes L_{b\backslash s}\ar[r]^{h_s\otimes id}&L_s\otimes L_{b\backslash s}\ar[dr]^{\prod_{f\in\Delta_b}H_f}& 
\\
V_{b,f}\ar[ur]^{i_{b,tr}}\ar[dr]^{i}& & &V_{b,f}
\\
&V_{s,f}\otimes V_{\overline{b\backslash s},f}\ar[r]^{h_{s,f}\otimes id}&V_{s,f}\otimes V_{\overline{b\backslash s},f}\ar[ur]^p&
}
$$
where the argument is the same as Lemma 3.

Similarly as Proposition 2, $V_{s,f}$ is spanned by $e_{i,j,k}$, where $i$ represents the conjugate class of $\pi_1(\Delta_s)\longrightarrow G$, $j$ represents holonomy map in $i$ and $k$ represents coloring for $j$. Since $\Delta_s$ is homeomorphic to a $n$-ball, $i,j$ are trivial denoted by $\unit$. Then $h_{s,f}$ can be written in matrix form $h_{s,f}(e_{\unit,\unit,k})=\sum_{k^{\prime}}a_k^{k^{\prime}}e_{\unit,\unit,k^{\prime}}$. For any $k^{\prime}$, since $\Delta_s$ is homeomorphic to a $n$-ball, there exist exactly $|G|$ $(g_1,...,g_{V_s})\in G^{V_s}$ such that 
$$(g_1\cdots g_{V_s})\cdot (e_{\unit,\unit,k}\otimes e)=e_{\unit,\unit,k^{\prime}}\otimes e_{k^{\prime}}
$$.
Any $x\in V_{b,tr}$ can be written as follows
\begin{align*}
x&=\frac{1}{|G|^{V_s}}\sum_{g_1,...,g_{V_s}}(g_1\cdots g_{V_s})\cdot(\bigotimes_{e_s\in\Delta_s}g_{e_s}\otimes\bigotimes_{e\in\Delta_{\overline{b\backslash s}}}g_e)
\\
&=\frac{1}{|G|^{V_s}}\sum_{g_1,...,g_{V_s}}(g_1\cdots g_{V_s})\cdot (e_{\unit,\unit,k}\otimes e)
\\
&=\frac{1}{|G|^{V_s}}\sum_{k^{\prime}}e_{\unit,\unit,k^{\prime}}\otimes\sum_{i=1}^{|G|} e_{k^{\prime},i}
\end{align*}
where the coloring is compatible with cocycle condition. Then we have
\begin{align*}
h_{b,tr}(x)&=\frac{1}{|G|^{V_s}}\sum_{k^{\prime}}\sum_{i=1}^{|G|}(\prod_{v\in\Delta_s}H_v)(p(h_{s,f}(e_{\unit,\unit,k^{\prime}})\otimes e_{k^{\prime},i}))
\\
&=\frac{1}{|G|^{V_s}}\sum_{k^{\prime}}\sum_{i=1}^{|G|}(\prod_{v\in\Delta_s}H_v)(p(\sum_{\tilde{k}}a_{k^{\prime}}^{\tilde{k}}e_{\unit,\unit,\tilde{k}}\otimes e_{k^{\prime},i}))
\\
&=\frac{1}{|G|^{V_s}}\sum_{k^{\prime}}\sum_{i=1}^{|G|}\sum_{\tilde{k}|_{\partial}=k^{\prime}|_{\partial}}a_{k^{\prime}}^{\tilde{k}}(\prod_{v\in\Delta_s}H_v)(e_{\unit,\unit,\tilde{k}}\otimes e_{k^{\prime},i})
\\
&=\frac{1}{|G|^{V_s}}\sum_{k^{\prime}}\sum_{i=1}^{|G|}\sum_{\tilde{k}|_{\partial}=k^{\prime}|_{\partial}}a_{k^{\prime}}^{\tilde{k}}\frac{1}{|G|^{V_s}}\sum_{\bar{k}}e_{\unit,\unit,\bar{k}}\otimes\sum_{j=1}^{|G|}e_{k^{\prime},i,\bar{k},j}
\\
&=(\frac{1}{|G|^{V_s-1}}\sum_{k^{\prime}}\sum_{\tilde{k}|_{\partial}=k^{\prime}|_{\partial}}a_{k^{\prime}}^{\tilde{k}})\cdot\frac{1}{|G|^{V_s}}\sum_{\bar{k}}\sum_{j=1}^{|G|}e_{\unit,\unit,\bar{k}}\otimes e_{\bar{k},j}
\end{align*}
We get $c=\frac{1}{|G|^{V_s-1}}\sum_{k^{\prime}}\sum_{\tilde{k}|_{\partial}=k^{\prime}|_{\partial}}a_{k^{\prime}}^{\tilde{k}}$ which does not depend on $x$.
\end{proof}

\printbibliography

\end{document}